\documentclass{llncs}
\usepackage{makeidx}  
\usepackage{amssymb,upgreek,nicefrac,amsmath}
\usepackage{sgame}
\usepackage[mathscr]{eucal}

\newcommand{\Cc}{C_b}
\newcommand{\cA}{{\mathcal{A}}}
\newcommand{\cZ}{{\mathcal{Z}}}
\newcommand{\cS}{{\mathcal{S}}}

\newcommand{\Bor}{{\mathfrak{B}}}
\newcommand{\cI}{\mathcal{I}}
\newcommand{\df}{\doteq}
\newcommand{\reward}[2]{u_{#1}(#2)}
\newcommand{\Prob}{{\mathbb{P}}} 
\newcommand{\Exp}{{\mathbb{E}}} 
\newcommand{\sF}{{\mathfrak{F}}}    

\newcommand{\supnorm}[1]{{\lVert}#1{\rVert}_{\infty}}

\newcommand{\QED}{\quad$\bullet$}

\newcommand{\tr}{^{\mathrm T}}

\newcommand{\magn}[1]{\left\vert #1 \right\vert}

\begin{document}
\title{Stochastic Stability Analysis of Perturbed Learning Automata with Constant Step-Size in Strategic-Form Games\thanks{This work has been partially supported by the European Union grant EU H2020-ICT-2014-1 project RePhrase (No. 644235).}}
\titlerunning{}  
%
\author{Georgios C. Chasparis}
\authorrunning{Georgios C. Chasparis} 
%
\tocauthor{Georgios C. Chasparis}
\institute{Software Competence Center Hagenberg GmbH, Softwarepark 21, A-4232 Hagenberg, Austria\\
\email{georgios.chasparis@scch.at},
}

\maketitle              

\begin{abstract}
This paper considers a class of reinforcement-learning that belongs to the family of Learning Automata and provides a stochastic-stability analysis in strategic-form games. For this class of dynamics, convergence to pure Nash equilibria has been demonstrated only for the fine class of potential games. Prior work primarily provides convergence properties of the dynamics through stochastic approximations, where the asymptotic behavior can be associated with the limit points of an ordinary-differential equation (ODE). However, analyzing global convergence through the ODE-approximation requires the existence of a Lyapunov or a potential function, which naturally restricts the applicabity of these algorithms to a fine class of games. To overcome these limitations, this paper introduces an alternative framework for analyzing stochastic-stability that is based upon an explicit characterization of the (unique) invariant probability measure of the induced Markov chain. 
\end{abstract}

\section{Introduction} \label{sec:Introduction}

Recently, multi-agent formulations have been utilized to tackle distributed optimization problems, since communication and computation complexity might be an issue in centralized optimization problems. 
In such formulations, decisions are usually taken in a repeated fashion, where agents select their next actions based on their \emph{own} prior experience of the game. 

The present paper discusses a class of reinforcement-learning dynamics, that belongs to the large family of Learning Automata \cite{Tsetlin73,Narendra89}, within the context of (non-cooperative) strategic-form games. 
In this class of dynamics, agents are repeatedly involved in a game with a fixed payoff-matrix, and they need to decide which action to play next having only access to their \emph{own} prior actions and payoffs. In Learning Automata, agents build their confidence over an action through repeated selection of this action and proportionally to the reward received from this action. Naturally, it has been utilized to analyze human-like (bounded) rationality \cite{Arthur93}.

Reinforcement learning has been applied in evolutionary economics, for modeling human and economic behavior \cite{Arthur93,BorgersSarin97,Erev98,HopkinsPosch05,Beggs05}. 
It is also highly attractive to several engineering applications, since agents do not need to know neither the actions of the other agents, nor their own utility function. It has been utilized for system identification and pattern recognition \cite{ThathacharSastry04}, distributed network formation and coordination problems \cite{ChasparisShamma11_DGA}.

In strategic-form games, the main goal is to derive conditions under which convergence to Nash equilibria can be achieved. In social sciences, deriving such conditions may be important for justifying emergence of certain social phenomena. 
In engineering, convergence to Nash equilibria may also be desirable in distributed optimization problems, when the set of optimal solutions coincides with the set of Nash equilibria. 

In Learning Automata, deriving conditions under which convergence to Nash equilibria is achieved may not be a trivial task. In particular, there are two main difficulties: a) excluding convergence to pure strategies that are \emph{not} Nash equilibria, and b) excluding convergence to mixed strategy profiles. As it will be discussed in detail in a forthcoming Section~\ref{sec:ReinforcementLearning}, for some classes of (discrete-time) reinforcement-learning algorithms, convergence to non-Nash pure strategies may be achieved with positive probability. Moreover, excluding convergence to mixed strategy profiles may only be achieved under strong conditions in the utilities of the agents, (e.g., existence of a potential function).

In the present paper, we consider a class of (discrete-time) reinforcement-learning algorithms introduced in \cite{ChasparisShamma11_DGA} that is closely related to existing algorithms for modeling human-like behavior, e.g., \cite{Arthur93}. The main difference with prior reinforcement learning schemes lies in a) the step-size sequence, and b) the perturbation (or \emph{mutations}) term. The step-size sequence is assumed constant, thus introducing a fading-memory effect of past experiences in each agent's strategy. On the other hand, the perturbation term introduces errors in the selection process of each agent. 
Both these two features can be used for designing a desirable asymptotic behavior.

We provide an analytical framework for deriving conclusions over the asymptotic behavior of the dynamics that is based on an explicit characterization of the invariant probability measure of the induced Markov chain. In particular, \emph{we show that in all strategic-form games satisfying the Positive-Utility Property, the support of the invariant probability measure coincides with the set of pure strategy profiles}. This extends prior work where nonconvergence to mixed strategy profiles may only be excluded under strong conditions in the payoff matrix (e.g., existence of a potential function). A detailed discussion of the exact contributions of this paper is provided in the forthcoming Section~\ref{sec:ReinforcementLearning}. At the end of the paper, we also provide a brief discussion over how the proposed framework can be further utilized to provide a more detailed characterization of the stochastically stable states (e.g., excluding convergence to non-Nash pure strategy profiles). Due to space limitations, this analysis is not presented in this paper.

In the remainder of the paper, Section~\ref{sec:ReinforcementLearning} presents a class of reinforcement-learning dynamics, related work and the main contribution of this paper. Section~\ref{sec:ConvergenceAnalysis} provides the main result of this paper (Theorem~\ref{Th:StochasticStability}), where the set of stochastically stable states is characterized. A short discussion is also provided over the significance of this result and how it can be utilized to provide further conclusions. Finally, Section~\ref{sec:TechnicalDerivation} provides the technical derivation of the main result and Section~\ref{sec:Conclusions} presents concluding remarks.

{\bf Notation:}
\begin{itemize}
\item For a Euclidean topological space $\mathcal{X}\subset\mathbb{R}^{n}$, let $\mathcal{N}_{\delta}(x)$ denote the $\delta$-neighborhood of $x\in\mathbb{R}^{n}$, i.e.,
\begin{equation*}
\mathcal{N}_{\delta}(x) \df \{y\in\mathcal{X}:|x-y|<\delta\},
\end{equation*}
where $|\cdot|$ denotes the Euclidean distance.
\item $e_j$ denotes the \emph{unit vector} in $\mathbb{R}^{n}$ where its $j$th entry is equal to 1 and all other entries is equal to 0.
\item $\Delta(n)$ denotes the \emph{probability simplex} of dimension $n$, i.e.,
\begin{equation*}
\Delta(n) \df \left\{ x\in\mathbb{R}^{n} : x\geq{0}, \mathbf{1}\tr x=1 \right\}.
\end{equation*}
\item For some set $A$ in a topological space $\cZ$, let $\mathbb{I}_{A}:\cZ\to\{0,1\}$ denote the index function, i.e.,
\begin{eqnarray*}
\mathbb{I}_{A}(x) \df \begin{cases}
1 & \mbox{ if } x\in{A}, \\
0 & \mbox{ else.}
\end{cases}
\end{eqnarray*}

\item $\delta_x$ denotes the Dirac measure at $x$.

\item Let $A$ be a finite set and let any (finite) probability distribution $\sigma\in\Delta(\magn{A})$. The random selection of an element of $A$ will be denoted ${\rm rand}_{\sigma}[A]$. If $\sigma=(\nicefrac{1}{\magn{A}},...,\nicefrac{1}{\magn{A}})$, i.e., it corresponds to the uniform distribution, the random selection will be denoted by ${\rm rand}_{\rm unif}[A]$.

\end{itemize}

\section{Reinforcement Learning}	\label{sec:ReinforcementLearning}

\subsection{Terminology}

We consider the standard setup of finite strategic-form games. Consider a finite set of agents (or \emph{players}) $\mathcal{I} = \{1,...,n\}$, and let each agent have a finite set of actions $\mathcal{A}_i$. Let $\alpha_i\in\mathcal{A}_i$ denote any such action of agent $i$. The set of \emph{action profiles} is the Cartesian product $\mathcal{A}\df\mathcal{A}_1\times\cdots\times\mathcal{A}_n$ and let $\alpha=(\alpha_1,...,\alpha_n)$ be a representative element of this set. We will denote $-i$ to be the complementary set $\cI\backslash{i}$ and often decompose an action profile as follows $\alpha=(\alpha_i,\alpha_{-i})$. The \emph{payoff/utility function} of agent $i$ is a mapping $\reward{i}{\cdot}:\mathcal{A}\to\mathbb{R}$. A \emph{strategic-form game} is defined by the triple $\langle{\mathcal{I},\mathcal{A},\{\reward{i}{\cdot}\}_i}\rangle$. 

\textit{\textbf{For the remainder of the paper}}, we will be concerned with strategic-form games that satisfy the \emph{\textbf{Positive-Utility Property}}. 

\begin{property}[Positive Utility Property]		\label{P:PositiveUtilityProperty}
For any agent $i\in\mathcal{I}$ and any action profile $\alpha\in\mathcal{A}$, $\reward{i}{\alpha}>0$.
\end{property}

\subsection{Reinforcement-learning algorithm}

We consider a form of reinforcement learning that belongs to the general class of \emph{learning automata} \cite{Narendra89}. In learning automata, each agent updates a finite probability distribution $x_i\in\Delta(\magn{\mathcal{A}_i})$ representing its beliefs with respect to the most profitable action. 
The precise manner in which $x_i(t)$ changes at time $t$, depending on the performed action and the response of the environment, completely defines the reinforcement learning model.

The proposed reinforcement learning model is described in Table~\ref{Tb:ReinforcementLearning}. At the first step, each agent $i$ updates its action given its current strategy vector $x_i(t)$. Its selection is slightly perturbed by a perturbation (or \emph{mutations}) factor $\lambda>0$, such that, with a small probability $\lambda$ agent $i$ follows a uniform strategy (or, it \emph{trembles}).  At the second step, agent $i$ evaluates its new selection by collecting a utility measurement, while in the last step, agent $i$ updates its strategy vector given its new experience.
\begin{table}[t!]
\boxed{
\begin{minipage}{0.98\textwidth}
At fixed time instances $t=1,2,...$, and for each agent $i\in\cI$, the following steps are executed recursively. Let $\alpha_i(t)$ and $x_i(t)$ denote the current action and strategy of agent $i$, respectively. 
\begin{enumerate}
\item (\emph{\textbf{action update}}) Agent $i\in\cI$ selects a new action $\alpha_i(t+1)$ as follows: 
\begin{eqnarray}	\label{eq:ActionUpdate}
\alpha_i(t+1) = \begin{cases}
{\rm rand}_{x_i(t)}[\mathcal{A}_i], & \mbox{ with probability } 1-\lambda, \\
{\rm rnad}_{\rm unif}[\mathcal{A}_i], & \mbox{ with probability } \lambda,
\end{cases}
\end{eqnarray} 
for some small perturbation factor $\lambda>0$.

\item (\emph{\textbf{evaluation}}) Agent $i$ applies its new action $\alpha_i(t+1)$ and retrieves a measurement of its utility function $\reward{i}{\alpha(t+1)}>0$. 

\item (\emph{\textbf{strategy update}}) Agent $i$ revises its strategy vector $x_i\in\Delta(\magn{\mathcal{A}_i})$ as follows: 
\begin{eqnarray}	\label{eq:ReinforcementLearningModel}
\lefteqn{x_i(t+1)} \cr & = & x_i(t) + \epsilon \cdot \reward{i}{\alpha(t+1)} \cdot [e_{\alpha_i(t+1)} - x_i(t)] \cr & \df & \mathcal{R}_{i}(\alpha(t+1),x_i(t)),
\end{eqnarray}
for some constant step size $\epsilon>0$.

\end{enumerate}
\end{minipage}
}
\caption{Perturbed Reinforcement Learning.}
\label{Tb:ReinforcementLearning}
\end{table}

Here we identify actions $\mathcal{A}_i$ with vertices of the simplex, $\{e_1,...,e_{\magn{\mathcal{A}_i}}\}$. For example, if agent $i$ selects its $j$th action at time $t$, then $e_{\alpha_i(t)}\equiv e_j$. Note that by letting the step-size $\epsilon$ to be sufficiently small and since the utility function $\reward{i}{\cdot}$ is uniformly bounded in $\mathcal{A}$, $x_i(t)\in\Delta(\magn{\mathcal{A}_i})$ for all $t$. 

In case $\lambda=0$, the above update recursion will be referred to as the \emph{unperturbed reinforcement learning}.

\subsection{Related work}

\subsubsection*{Erev-Roth type dynamics}

In prior reinforcement learning in games, analysis has been restricted to decreasing step-size sequences $\epsilon(t)$ and $\lambda=0$. More specifically, in \cite{Arthur93}, the step-size sequence of agent $i$ is $\epsilon_i(t) = 1/(ct^{\nu}+\reward{i}{\alpha(t+1)}$ for some positive constant $c$ and for $0<\nu<1$ (in the place of the constant step size $\epsilon$ of (\ref{eq:ReinforcementLearningModel})). A comparative model is also used by \cite{HopkinsPosch05}, with $\epsilon_i(t) = 1/(V_i(t)+\reward{i}{\alpha(t+1)})$, where $V_i(t)$ is the accumulated benefits of agent $i$ up to time $t$ which gives rise to an urn process 
\cite{Erev98}. Some similarities are also shared with the Cross' learning model of \cite{BorgersSarin97}, where $\epsilon(t)=1$ and $\reward{i}{\alpha(t)}\leq{1}$, and its modification presented in \cite{Leslie04}, where $\epsilon(t)$, instead, is assumed decreasing.

The main difference of the proposed reinforcement-learning algorithm (Table~\ref{Tb:ReinforcementLearning}) lies in the perturbation parameter $\lambda>0$ which was first introduced and analyzed in \cite{ChasparisShamma11_DGA}. A state-dependent perturbation term has also been investigated in \cite{ChasparisShammaRantzer15}. The perturbation parameter may serve as an equilibrium selection mechanism, since \emph{it excludes convergence to non-Nash action profiles}. It resolved one of the main issues of several (discrete-time) reinforcement-learning algorithms, that is the positive probability of convergence to non-Nash action profiles under some conditions in the payoff function and the step-size sequence. 

This issue has also been raised by \cite{Posch97,HopkinsPosch05}. Reference \cite{Posch97} considered the model by \cite{Arthur93} and showed that convergence to non-Nash pure strategy profiles can be excluded as long as $c>\reward{i}{\alpha}$ for all $i\in\mathcal{I}$ and $\nu=1$. On the other hand, convergence to non-Nash action profiles was not an issue with the urn model of \cite{Erev98} (as analyzed in \cite{HopkinsPosch05}). However, the use of an urn-process type step-size sequence significantly reduces the applicability of the reinforcement learning scheme. In conclusion, the perturbation parameter $\lambda>0$ may serve as a design tool for reinforcing convergence to Nash equilibria without necessarily employing an urn-process type step-size sequence. For engineering applications this is a desirable feature.

Although excluding convergence to non-Nash pure strategies can be guaranteed by using $\lambda>0$, establishing convergence to pure Nash equilibria may still be an issue, since it further requires excluding convergence to mixed strategy profiles. As presented in \cite{ChasparisShammaRantzer15}, this can be guaranteed only under strong conditions in the payoff matrix. For example, as shown in \cite[Proposition~8]{ChasparisShammaRantzer15}, excluding convergence to mixed strategy profiles requires a) the existence of a potential function, b) conditions over the second gradient of the potential function. Requiring the existence of a potential function considerably restricts the class of games where equilibrium selection can be described. Furthermore, condition (b) may not easily be verified in games of large number of players or actions.

\subsubsection*{Learning automata}

Certain forms of learning automata have been shown to converge to Nash equilibria in some classes of strategic-form games. For example, in \cite{Narendra89}, and for a generalized nonlinear reward-inaction scheme, convergence to Nash equilibrium strategies can be shown in identical interest games. Similar are the results presented in \cite{Sastry94} for a linear reward-inaction scheme. These convergence results are restricted to games of payoffs in $[0,1]$. Extension to a larger class of games is possible if \emph{absolute monotonicity} (cf.,~\cite[Definition~8.1]{Narendra89}) is shown (similarly to the discussion in \cite[Proposition~8]{ChasparisShammaRantzer15}). 
%

Reference~\cite{verbeeck_exploring_2007} introduced a class of linear reward-inaction schemes in combination with a coordinated exploration phase so that convergence to the efficient Nash equilibrium is achieved. However, coordination of the exploration phase requires communication between the players.

Recently, work by the author \cite{ChasparisShamma11_DGA} has introduced a new class of learning automata (namely, perturbed learning automata) which can be applied in games with no restriction in the payoff matrix. Furthermore, a small perturbation factor also influences the decisions of the players, through which convergence to non-Nash pure strategy profiles can be excluded. However, to demonstrate global convergence, a monotonicity condition still needs to be established \cite{ChasparisShammaRantzer15}.

\subsubsection*{$Q$-learning}

Similar questions of convergence to Nash equilibria also appear in alternative reinforcement learning formulations, such as approximate dynamic programming methodologies and $Q$-learning. However, this is usually accomplished under a stronger set of assumptions, which increases the computational complexity of the dynamics. For example, the Nash-Q learning algorithm of \cite{hu_nash_2003} addresses the problem of maximizing the discounted expected rewards for each agent by updating an approximation of the cost-to-go function (or $Q$-values). Alternative objectives may be used, such as the minimax criterion of \cite{littman_markov_1994}. However, it is indirectly assumed that agents need to have full access to the joint action space and the rewards received by the other agents. 

More recently, reference \cite{chapman_convergent_2013} introduces a $Q$-learning scheme in combination with either adaptive play or better-reply dynamics in order to attain convergence to Nash equilibria in potential games \cite{MondererShapley96} or weakly-acyclic games. However, this form of dynamics require that each player observes the actions selected by the other players, since a $Q$-value needs to be assigned in each joint action.

When the evaluation of the $Q$-values is totally independent, as in the individual $Q$-learning in \cite{leslie_individual_2005}, then convergence to Nash equilibria has been shown only for 2-player zero-sum games and 2-player partnership games with countably many Nash equilibria. Currently, there are no convergent results in games in multi-player games.

\subsubsection*{Payoff-based learning}

The aforementioned types of dynamics can be considered as a form of payoff-based learning dynamics, since adaptation is only governed by the perceived utility of the players. Recently, there have been several attempts for establishing convergence to Nash equilibria through alternative payoff-based learning dynamics, (see, e.g., the benchmark-based dynamics of \cite{marden_payoff_2009}, or the aspiration-based dynamics in \cite{ChasparisAriShamma13_SIAM}). For these type of dynamics, convergence to Nash equilibria can be established without requiring any strong monotonicity property (e.g., in multi-player weakly-acyclic games in \cite{marden_payoff_2009}). However, an investigation is required with respect to the resulting convergence rates as compared to the dynamics incorporating policy iterations (e.g., the Erev-Roth type of dynamics or the learning automata discussed above).

\subsection{Objective}

This paper provides an analytical framework for analyzing convergence in multi-player strategic-form games when players implement a class of perturbed learning-automata. We wish to impose no strong monotonicity assumptions in the structure of the game (e.g., the existence of a potential function). We provide a characterization of the invariant probability measure of the induced Markov chain that shows that only the pure-strategy profiles belong to its support. Thus, we implicitly exclude convergence to any mixed strategy profile (including mixed Nash equilibria). This result imposes no restrictions in the payoff matrix other than the Positive-Utility Property.

%

\section{Convergence Analysis}	\label{sec:ConvergenceAnalysis}

\subsection{Terminology and notation}	\label{sec:Terminology}

Let $\cZ\df \mathcal{A}\times \mathbf{\Delta}$, where $\mathbf{\Delta}\df\Delta(\magn{\mathcal{A}_1})\times\ldots\times\Delta(\magn{\mathcal{A}_n})$, i.e., pairs of joint actions $\alpha$ and nominal strategy profiles $x$. The set $\mathcal{A}$ is endowed with the discrete topology, $\mathbf{\Delta}$ with its usual Euclidean topology, and $\cZ$ with the corresponding product topology. We also let $\Bor(\cZ)$ denote the Borel $\sigma$-field of $\cZ$, and $\mathfrak{P}(\cZ)$ the set of probability measures on $\Bor(\cZ)$ endowed with the Prohorov topology, i.e., the topology of weak convergence. The algorithm introduced in Table~\ref{Tb:ReinforcementLearning} defines an $\cZ$-valued Markov chain. Let $P_{\lambda}:\cZ\times\Bor(\cZ)\to[0,1]$ denote its transition probability function (t.p.f.), parameterized by $\lambda>0$. We refer to the process with $\lambda>0$ as the \emph{perturbed process}. Let also $P:\cZ\times\Bor(\cZ)$ denote the t.p.f. of the \emph{unperturbed process}, i.e., when $\lambda=0$.

We let $C_b(\cZ)$ denote the Banach space of real-valued continuous functions on $\cZ$ under the sup-norm (denoted by $\|\cdot\|_{\infty}$) topology. For $f\in\Cc(\cZ)$, define
\begin{equation*}
P_{\lambda}f(z) \df \int_{\cZ}P_{\lambda}(z,dy)f(y),
\end{equation*}
and 
\begin{equation*}
\mu[f] \df \int_{\cZ}\mu(dx)f(z), \mbox{ for } \mu\in\mathfrak{P}(\cZ).
\end{equation*}

The process governed by the unperturbed process $P$ will be denoted by $\{Z_{t} : t\ge0\}$. Let $\Omega\df\cZ^{\infty}$ denote the canonical path space, i.e., an element $\omega\in\Omega$ is a sequence $\{\omega(0),\omega(1),\dotsc\}$, with $\omega(t)= (\alpha(t),x(t))\in\cZ$. We use the same notation for the elements $(\alpha,x)$ of the space $\cZ$ and for the coordinates of the process $Z_{t}=(\alpha(t),x(t))$.
Let also $\Prob_{z}[\cdot]$ denote the unique probability measure induced by the unperturbed process $P$ on the product $\sigma$-algebra of $\cZ^{\infty}$, initialized at $z=(\alpha,x)$, and $\Exp_{z}[\cdot]$ the corresponding expectation operator. Let also $\sF_{t}$, 
$t\geq{0}$, denote the $\sigma$-algebra generated by $\{Z_{\tau},~\tau\le{t}\}$.

\subsection{Stochastic stability}	\label{sec:StochasticStability}

First, we note that both $P$ and $P_{\lambda}$ ($\lambda>0$) satisfy the \emph{weak Feller property} (cf.,~\cite[Definition~4.4.2]{Lerma03}).
\begin{proposition}		\label{Pr:WeakFeller}
Both the unperturbed process $P$ ($\lambda=0$) and the perturbed process $P_{\lambda}$ ($\lambda>0$) have the weak Feller property.
\end{proposition}
\begin{proof}
Let us consider any sequence $\{Z^{(k)}=(\alpha^{(k)},x^{(k)})\}$ such that $Z^{(k)}\to{Z}=(\alpha,x)\in\cZ$.

For the unperturbed process governed by $P(\cdot,\cdot)$, and for any open set $O\in\Bor(\cZ)$, the following holds:
\begin{eqnarray*}
\lefteqn{P(Z^{(k)}=(\alpha^{(k)},x^{(k)}),O) }\cr 
& = & \sum_{\alpha\in\mathcal{P}_{\cA}(O)}\Big\{\Prob_{Z^{(k)}}[{\rm rand}_{x_i^{(k)}}[\mathcal{A}_i]=\alpha_i, \forall i\in\mathcal{I}] \Big. \cdot \cr && \Big.\prod_{i=1}^{n}\Prob_{Z^{(k)}}[\mathcal{R}_{i}(\alpha,x_i^{(k)})\in\mathcal{P}_{\mathcal{X}_i}(O)] \Big\} \cr
& = & \sum_{\alpha\in\mathcal{P}_{\mathcal{A}}(O)} \Big\{\prod_{i=1}^{n}\mathbb{I}_{\mathcal{P}_{\mathcal{X}_i}(O)}(\mathcal{R}_i(\alpha,x_i^{(k)})) x_{i\alpha_i}^{(k)}\Big\},
\end{eqnarray*}
where $\mathcal{P}_{\mathcal{X}_i}(O)$ and $\mathcal{P}_{\mathcal{A}}(O)$ are the \emph{canonical projections} defined by the product topology. Similarly, we have: 
\begin{eqnarray*}
\lefteqn{P(Z=(\alpha,x),O) } \cr 
& = & \sum_{\alpha\in\mathcal{P}_{\mathcal{A}}(O)} \Big\{\prod_{i=1}^{n}\mathbb{I}_{\mathcal{P}_{\mathcal{X}_i}(O)}\left(\mathcal{R}_i\left(\alpha,x_i\right)\right) x_{i\alpha_i}\Big\}.
\end{eqnarray*}

(a) \emph{Consider the case $x\in\mathbf{\Delta}^o$,} i.e., $x$ belongs to the interior of $\mathbf{\Delta}$. For all $i\in\cI$, due to the continuity of $\mathcal{R}_i(\cdot,\cdot)$ with respect to its second argument, and the fact that $O$ is an open set, there exists $\delta>0$ such that $\mathbb{I}_{\mathcal{P}_{\mathcal{X}_i}(O)}(\mathcal{R}_i(\alpha,x_i)) = \mathbb{I}_{\mathcal{P}_{\mathcal{X}_i}(O)}(\alpha,y_i))$ for all $y_i\in\mathcal{N}_{\delta}(x_i)$. Thus, for any sequence  $Z^{(k)}=(\alpha^{(k)},x^{(k)})$ such that $Z^{(k)}\to Z=(\alpha,x)$, we have that 
$P(Z^{(k)},O) \to P(Z,O),$ as $k\to\infty$.

(b) \emph{Consider the case $x\in\partial\mathbf{\Delta}$,} i.e., $x$ belongs to the boundary of $\mathbf{\Delta}$.  Then, there exists $i\in\cI$ such that $x_i\in\partial\Delta(\magn{\mathcal{A}_i})$, i.e., there exists an action $j\in\cA_i$ such that $x_{ij}=0$. For any open set $O\in\Bor(\cZ)$, $x_i\notin\mathcal{P}_{\mathcal{X}_i}(O)$. Furthermore, for any $\alpha_i\in {\rm rand}_{x_i}[\cA_i]$, $\mathbb{I}_{\mathcal{P}_{\mathcal{X}_i}(O)}(\mathcal{R}_i((\alpha_i,\alpha_{-i}),x_i) = 0$ (since $x_{ij}=0$ and therefore $x_i$ cannot escape from the boundary). This directly implies that $P(Z=(\alpha,x),O)=0$. Construct a sequence $(\alpha^{(k)},x^{(k)})$ that converges to $(\alpha,x)$ such that $\alpha^{(k)}=\alpha$, $x_{i\alpha_i}^{(k)}>0$ and $x_{i}=e_{\alpha_i}$, i.e., the strategy of player $i$ converges to the vertex of action $\alpha_i$. Pick also $O\in\Bor(\cZ)$, such that $\mathbb{I}_{\mathcal{P}_{\mathcal{X}_i}(O)}(\mathcal{R}_i(\alpha,x_i^{(k)})) = 1$ for all large $k$. This is always possible by selecting an open set $O$ such that $x\in\partial\mathcal{P}_{\mathcal{X}}(O)$ and $x^{(k)}\in\mathcal{P}_{\mathcal{X}}(O)$ for all $k$. In this case, $\lim_{k\to\infty}P(Z^{(k)},O) = 1$. We conclude that for any sequence $Z^{(k)}=(\alpha^{(k)},x^{(k)})$ that converges to $Z=(\alpha,x)$, such that $x\in\partial\mathbf{\Delta}$, and for any open set $O\in\Bor(\cZ)$, $$\lim_{k\to\infty}P(Z^{(k)},O) \geq P(Z,O)=0.$$

By \cite[Proposition~7.2.1]{Lerma03}, we conclude that $P$ satisfies the weak Feller property. The same steps can be followed to show that $P_{\lambda}$ also satisfies the weak Feller property. \QED
\end{proof}

The measure $\mu_{\lambda}\in\mathfrak{P}(\cZ)$ is called an \emph{invariant probability measure} for $P_{\lambda}$ if
\begin{equation*}
(\mu_{\lambda}P_{\lambda})(A) \df \int_{\cZ}\mu_{\lambda}(dx)P_{\lambda}(z,A) = \mu_{\lambda}(A), \qquad A\in\Bor(\cZ).
\end{equation*}
Since $\cZ$ defines a locally compact separable metric space and $P$, $P_{\lambda}$ have the weak Feller property, they both admit an invariant probability measure, denoted $\mu$ and $\mu_{\lambda}$, respectively \cite[Theorem~7.2.3]{Lerma03}.

We would like to characterize the \emph{stochastically stable states} $z\in\cZ$ of $P_{\lambda}$, that is any state $z\in\cZ$ for which any collection of invariant probability measures $\{\mu_{\lambda}\in\mathfrak{P}(\cZ):\mu_{\lambda}P_{\lambda}=\mu_{\lambda},\lambda>0\}$ satisfies $\liminf_{\lambda\to{0}}\mu_{\lambda}(z)>0$. As the forthcoming analysis will show, the stochastically stable states will be a subset of the set of \emph{pure strategy states} (p.s.s.) defined as follows:
\begin{definition}[Pure Strategy State]	\label{def:PureStrategyState}
\textit{
A pure strategy state is a state $s=(\alpha,x)\in\cZ$ such that for all $i\in\mathcal{I}$, $x_i = e_{\alpha_i}$, i.e., $x_i$ coincides with the vertex of the probability simplex $\Delta(\magn{\mathcal{A}_i})$ which assigns probability 1 to action $\alpha_i$.
}
\end{definition}

We will denote the set of pure strategy states by $\mathcal{S}$.

\begin{theorem}[Stochastic Stability]		\label{Th:StochasticStability}
There exists a unique probability vector $\pi=(\pi_1,...,\pi_{\magn{\mathcal{S}}})$ such that for any collection of invariant probability measures $\{\mu_{\lambda}\in\mathfrak{P}(\cZ):\mu_{\lambda}P_{\lambda}=\mu_{\lambda}, \lambda>0\}$, the following hold:
\begin{itemize}
\item[(a)] $\lim_{\lambda\to{0}}\mu_{\lambda}(\cdot) = \hat{\mu}(\cdot) \df \sum_{s\in\mathcal{S}}\pi_s\delta_s(\cdot),$ where convergence is in the weak sense.
\item[(b)] The probability vector $\pi$ is an invariant distribution of the (finite-state) Markov process $\hat{P}$, such that, for any $s,s'\in\mathcal{S}$, 
\begin{equation}	\label{eq:FiniteStateMarkovChain}
\hat{P}_{ss'} \df \lim_{t\to\infty} QP^t(s,\mathcal{N}_{\delta}(s')),
\end{equation}
for any $\delta>0$ sufficiently small, where $Q$ is the t.p.f. corresponding to \emph{only one player trembling} (i.e., following the uniform distribution of (\ref{eq:ActionUpdate})).
\end{itemize}
\end{theorem}

The proof of Theorem~\ref{Th:StochasticStability} requires a series of propositions and will be presented in detail in Section~\ref{sec:TechnicalDerivation}. 

\subsection{Discussion}

Theorem~\ref{Th:StochasticStability} establishes an important observation. That is, the ``\textit{equivalence}'' (in a weak convergence sense) of the original (perturbed) learning process with a simplified process, where \emph{agents simultaneously tremble at the first iteration and then they do not tremble}. This form of simplification of the dynamics has originally been exploited to analyze \emph{aspiration learning} dynamics in \cite{ChasparisAriShamma13_SIAM}, and it is based upon the fact that \emph{under the unperturbed dynamics, agents' strategies will eventually converge to a pure strategy profile}.

Furthermore, the limiting behavior of the original (perturbed) dynamics can be characterized by the (\emph{unique}) invariant distribution of a finite-state Markov chain $\{P_{ss'}\}$, whose states correspond to the pure-strategy states of the game. In other words, \emph{we should expect that as the perturbation parameter $\lambda$ approaches zero, the algorithm spends the majority of the time on pure strategy profiles}. The importance of this result lies on the fact that no constraints have been imposed in the payoff matrix of the game other than the Positive-Utility Property \ref{P:PositiveUtilityProperty}. Thus, it extends to games beyond the fine set of potential games.

This convergence result can further be augmented with an ODE analysis for stochastic approximations to exclude convergence to pure strategies that are not Nash equilibria (as derived in \cite{ChasparisShammaRantzer15} for the case of diminishing step size). Due to space limitations this analysis is not presented in this paper, however it can be the subject of future work.

\section{Technical Derivation}	\label{sec:TechnicalDerivation}


\subsection{Unperturbed Process}	\label{Sc:UnperturbedProcess}

For $t\ge0$ define the sets
\begin{align*}
A_{t} &\df \left\{\omega\in\Omega:\alpha(\tau)=\alpha(t)\,, \text{~for all~} \tau \geq t \right\}\,,\\[5pt]
B_{t} &\df \{\omega\in\Omega:\alpha(\tau)=\alpha(0)\,, \text{~for all~} 0\le\tau\le{t}\}\,.
\end{align*}
Note that $\{B_{t}:t\ge0\}$ is a non-increasing sequence, i.e., $B_{t+1}\subseteq B_{t}$, while $\{A_{t}:t\ge0\}$ is non-decreasing, i.e., $A_{t+1}\supseteq A_t$. Let
\begin{equation*}
A_{\infty} \df\bigcup_{t=0}^{\infty}A_{t} \mbox{ and } B_{\infty}\df\bigcap_{t=1}^{\infty}B_{t}.
\end{equation*}
In other words, \emph{the set $A_{\infty}$ corresponds to the event that agents eventually play the same action profile, while $B_{\infty}$ corresponds to the event that agents never change their actions}. 

\begin{proposition}[Convergence to p.s.s.]	\label{Pr:ConvergenceToPSS}
Let us assume that the step size $\epsilon>0$ is sufficiently small such that $0<\epsilon u_i(\alpha)<1$ for all $\alpha\in\mathcal{A}$ and for all agents $i\in\mathcal{I}$. Then, the following hold: 
\begin{itemize}
\item[(a)] $\inf_{z\in\cZ}\;\mathbb{P}_{z}[B_{\infty}]>0$, 
\item[(b)] $\inf_{z\in\cZ}\Prob_{z}[A_\infty]=1$.
\end{itemize}
\end{proposition}

The first statement of Proposition~\ref{Pr:ConvergenceToPSS} states that \emph{the probability that agents never change their actions is bounded away from zero}, while the second statement states that \emph{the probability that eventually agents play the same action profile is one}.
\begin{proof}
(a) Let us consider an action profile $\alpha=(\alpha_1,...,\alpha_n)\in\mathcal{A}$, and an initial strategy profile $x(0)=(x_1(0),...,x_n(0))$ such that $x_{i\alpha_i}(0)>0$ for all $i\in\mathcal{I}$. Note that if the same action profile $\alpha$ is selected up to time $t$, then the strategy of agent $i$ satisfies:
\begin{equation}	\label{eq:ConvergenceToPSS:AccumulatedStrategy}
x_{i}(t) = e_{\alpha_i} - (1-\epsilon u_i(\alpha))^{t}(e_{\alpha_i}-x_i(0)).
\end{equation}
Given that $B_t$ is non-increasing, from continuity from above we have
\begin{equation*}
\Prob_{z}[B_{\infty}] = \lim_{t\to\infty}\Prob_{z}[B_t] = \lim_{t\to\infty}\prod_{k=0}^{t}\prod_{i=1}^{n}x_{i\alpha_i}(k).
\end{equation*}
Note that $\Prob[B_{\infty}] > 0$ if and only if 
\begin{equation}	\label{eq:ConvergenceToPSS:Condition1}
\sum_{t=1}^{\infty}\log(x_{i\alpha_i}(t)) > -\infty.
\end{equation}
Let us introduce the variable $$y_i(t) \df 1-x_{i\alpha_i}(t) = \sum_{j\in\mathcal{A}_i\backslash\alpha_i}x_{ij}(t),$$ which corresponds to the probability of agent $i$ selecting any action other than $\alpha_i$. Condition (\ref{eq:ConvergenceToPSS:Condition1}) is equivalent to
\begin{equation}	\label{eq:ConvergenceToPSS:Condition2}
-\sum_{t=0}^{\infty}\log(1-y_i(t)) < \infty,	\mbox{ for all } i\in\mathcal{I}.
\end{equation}
We also have that 
\begin{equation*}
\lim_{t\to\infty}\frac{-\log(1-y_i(t))}{y_i(t)} = \lim_{t\to\infty}\frac{1}{1-y_i(t)} > \rho
\end{equation*}
for some $\rho>0$, since $0\leq y_i(t) \leq 1$. Thus, from the Limit Comparison Test, we conclude that condition (\ref{eq:ConvergenceToPSS:Condition2}) holds if and only if $\sum_{t=1}^{\infty}y_i(t) < \infty$, for each $i\in\mathcal{I}$.

Lastly, note that $y_{i}(t+1)/y_i(t) = 1-\epsilon u_i(\alpha)$. By Raabe's criterion, the series $\sum_{t=0}^{\infty}y_i(t)$ is convergent if $\lim_{t\to\infty}t\left(\nicefrac{y_i(t)}{y_i(t+1)} - 1 \right) > 1.$
We have
\begin{eqnarray*}
t\left(\frac{y_i(t)}{y_i(t+1)}-1\right) = t\frac{\epsilon u_i(\alpha)}{1-\epsilon u_i(\alpha)}.
\end{eqnarray*}
Thus, if $\epsilon u_i(\alpha)<1$ for all $\alpha\in\mathcal{A}$ and $i\in\mathcal{I}$, then $1-\epsilon u_i(\alpha)>0$ and 
$\lim_{t\to\infty}t(\nicefrac{\epsilon u_i(\alpha)}{1-\epsilon u_i(\alpha)}) > 1,$
which implies that the series $\sum_{t=1}^{\infty}y_i(t)$ is convergent. Thus, we conclude that $\Prob_z[B_{\infty}]>0$.

(b) Define the event
\begin{equation*}
C_t \df \left\{\exists\alpha'\neq\alpha(t):x_{i\alpha_i'}(t)>0, \mbox{ for all } i\in\mathcal{I}\right\},
\end{equation*}
i.e., $C_t$ corresponds to the event that there exists an action profile different from the current action profile for which the nominal strategy assigns positive probability for all agents $i$. Note that $A_t^c \subseteq C_t$, since $A_t^c$ occurs only if there is some action profile $\alpha'\neq\alpha(t)$ for which the nominal strategy assigns positive probability. This further implies that $\Prob_{z}[A_{t}^{c}]\leq\Prob_{z}[C_t]$. Then, we have:
\begin{eqnarray*}
\lefteqn{\Prob_{z}[A_{t+1}|A_t^{c}]} \cr & = & \frac{\Prob_{z}[A_{t+1}\cap A_{t}^{c}]}{\Prob_z[A_t^c]} \cr & \geq & \frac{\Prob_z[A_{t+1}\cap A_t^c]}{\Prob_z[C_t]} \cr
& \geq & \Prob_z[A_{t+1}\cap A_{t}^{c}|C_t] \cr
& = & \Prob_z[\{\alpha(\tau)=\alpha'\neq\alpha(t),\forall \tau>t\}|C_t] \cr
& \geq & \inf_{\alpha'\neq{\alpha}}\prod_{i=1}^{n}x_{i\alpha_i'}(t)\prod_{k=t+1}^{\infty}\left\{1-(1-\epsilon\reward{i}{\alpha'})^{k-t-1}c_i(\alpha')\right\} \cr
& \geq & \inf_{\alpha'\neq\alpha}\prod_{i=1}^{n}x_{i\alpha_i'}(t)\prod_{k=0}^{\infty}\left\{1-(1-\epsilon\reward{i}{\alpha'})^{k}c_i(\alpha')\right\}
\end{eqnarray*}
where $c_i(\alpha')\df 1 - x_{i\alpha_i'}(t) \geq 0$. We have already shown in part (a) that the second part of the r.h.s. is bounded away from zero. Therefore, we conclude that $\Prob_z[A_{t+1}|A_t^{c}]>0$. Thus, from the counterpart of the Borel-Cantelli Lemma, $\Prob_z\left[A_{\infty}\right]=1.$ \QED
\end{proof}

The above proposition is rather useful in characterizing the support of any invariant measure of the unperturbed process, as the following proposition shows.

\begin{proposition}[Limiting t.p.f. of unperturbed process]	\label{Pr:LimitingUnperturbedTPF}
Let $\mu$ denote an invariant probability measure of $P$. Then, there exists a t.p.f. $\Pi$ on $\cZ\times\Bor(\cZ)$ such that 
\begin{itemize}
\item[(a)] for $\mu$-a.e. $z\in\cZ$, $\Pi(z,\cdot)$ is an invariant probability measure for $P$;
\item[(b)] for all $f\in\Cc(\cZ)$, $\lim_{t\to\infty}\|P^tf-\Pi f\|_{\infty}=0$;
\item[(c)] $\mu$ is an invariant probability measure of $\Pi$;
\item[(d)] the support\footnote{The \emph{support} of a measure $\mu$ on $\cZ$ is the unique closed set $F\subset\Bor(\cZ)$ such that $\mu(\cZ\backslash{F})=0$ and $\mu(F\cap{O})>0$ for every open set $O\subset\cZ$ such that $F\cap{O}\neq\varnothing$.} of $\Pi$ is on $\mathcal{S}$ for all $z\in\cZ$.
\end{itemize}
\end{proposition}
\begin{proof}
The state space $\cZ$ is a locally compact separable metric space and the t.p.f. of the unperturbed process $P$ admits an invariant probability measure due to Proposition~\ref{Pr:WeakFeller}. Thus, statements (a), (b) and (c) follow directly from \cite[Theorem~5.2.2 (a), (b), (e)]{Lerma03}. 

(d) Let us assume that the support of $\Pi$ includes points in $\cZ$ other than the pure strategy states. Let also $O\subset\cZ$ be an open set such that $O\cap\cS=\varnothing$ and $\Pi(z^*,O)>0$ for some $z^*\in\cZ$. Given that $P^{t}$ converges weakly to $\Pi$ as $t\to\infty$, from Portmanteau theorem (cf.,~\cite[Theorem~1.4.16]{Lerma03}), we have that $$\liminf_{t\to\infty} P^{t}(z^*,O) \geq \Pi(z^*,O)>0.$$ This is a contradiction of Proposition~\ref{Pr:ConvergenceToPSS}(b). Thus, the conclusion follows. \QED
%
\end{proof}

Proposition~\ref{Pr:LimitingUnperturbedTPF} states that the limiting unperturbed t.p.f. converges weakly to a t.p.f. $\Pi$ which accepts the same invariant p.m. as $P$. Furthermore, \emph{the support of $\Pi$ is the set of pure strategy states $\mathcal{S}$}. This is a rather important observation, since the limiting perturbed process can also be ``related'' (in a weak-convergence sense) to the t.p.f. $\Pi$, as it will be shown in the following section.

\subsection{Decomposition of perturbed t.p.f.}		\label{sec:DecompositionPerturbedTPF}

We can decompose the t.p.f. of the perturbed process as follows:
$$P_{\lambda} = (1-\varphi(\lambda))P + \varphi(\lambda)Q_{\lambda}$$ where $\varphi(\lambda) = 1-(1-\lambda)^{n}$ is the probability that at least one agent trembles (since $(1-\lambda)^{n}$ is the probability that no agent trembles), and $Q_{\lambda}$ corresponds to the t.p.f. induced by the one-step reinforcement-learning update when at least one agent trembles. Note that $\varphi(\lambda)\to{0}$ as $\lambda\to{0}$. 

Define also $Q$ to be the t.p.f. when \emph{only one} players trembles, and $Q^*$ is the t.p.f. where at least two players tremble. Then, we may write:
\begin{equation}
Q_{\lambda} = (1-\psi(\lambda))Q + \psi(\lambda)Q^*,
\end{equation}
where $\psi(\lambda) \df 1-\frac{n\lambda}{1-(1-\lambda)^{n}}$ corresponds to the probability that at least two players tremble given that at least one player trembles. 

Let us also define the infinite-step t.p.f. when trembling only at the first step (briefly, \emph{lifted} t.p.f.) as follows: 
\begin{equation}
P_{\lambda}^{L} \df \varphi(\lambda)\sum_{t=0}^{\infty}(1-\varphi(\lambda))^{t}Q_{\lambda}P^{t} = Q_{\lambda}R_{\lambda}
\end{equation}
where
$R_{\lambda} \df \varphi(\lambda)\sum_{t=0}^{\infty}(1-\varphi(\lambda))^{t}P^{t},$
i.e., $R_{\lambda}$ corresponds to the \emph{resolvent} t.p.f.

\begin{proposition}[Invariant p.m. of perturbed process]		\label{Pr:WeakLimitPointsOfPerturbedInvariantMeasures}
The following hold:
\begin{itemize}
\item[(a)] For $f\in\Cc(\cZ)$, $\lim_{\lambda\to{0}}\|R_{\lambda}f-\Pi{f}\|_{\infty} = 0.$
\item[(b)] For $f\in\Cc(\cZ)$, $\lim_{\lambda\to{0}}\|P_{\lambda}^{L}f-Q\Pi{f}\|_{\infty} = 0$.
\item[(c)] Any invariant distribution $\mu_{\lambda}$ of $P_\lambda$ is also an invariant distribution of $P_{\lambda}^{L}$.
\item[(d)] Any weak limit point in $\mathfrak{P}(\cZ)$ of $\mu_{\lambda}$, as $\lambda\to{0}$, is an invariant probability measure of $Q\Pi$.
\end{itemize}
\end{proposition}
\begin{proof}
(a) For any $f\in C_b(\cZ)$, we have
\begin{eqnarray*}
\lefteqn{\|R_{\lambda}f - \Pi{f}\|_{\infty}} \cr 
& = & \supnorm{\varphi(\lambda)\sum_{t=0}^{\infty}(1-\varphi(\lambda))^tP^tf - \Pi f}  \cr 
& = & \supnorm{\varphi(\lambda)\sum_{t=0}^{\infty}(1-\varphi(\lambda))^t(P^t f - \Pi f)} \cr
\end{eqnarray*}
where we have used the property $\varphi(\lambda)\sum_{t=0}^{\infty}(1-\varphi(\lambda))^t=1$. Note that
\begin{eqnarray*}
\lefteqn{\varphi(\lambda)\sum_{t=T}^{\infty}(1-\varphi(\lambda))^t\supnorm{P^tf-\Pi f} } \cr
& \leq & (1-\varphi(\lambda))^{T}\sup_{t\geq{T}}\supnorm{P^tf - \Pi f}. 
\end{eqnarray*}
From Proposition~\ref{Pr:LimitingUnperturbedTPF}(b), we have that for any $\delta>0$, there exists $T=T(\delta)>0$ such that the r.h.s. is uniformly bounded by $\delta$ for all $t\geq T$. Thus, the sequence $$A_T\df\varphi(\lambda)\sum_{t=0}^{T}(1-\varphi(\lambda))^t(P^tf - \Pi f)$$ is Cauchy and therefore convergent (under the sup-norm). In other words, there exists $A\in\mathbb{R}$ such that $$\lim_{T\to\infty}\supnorm{A_{T}-A}=0.$$ 
For every $T>0$, we have
\begin{eqnarray*}
\supnorm{R_{\lambda}f-\Pi{f}} \leq \supnorm{A_{T}} + \supnorm{A - A_T}.
\end{eqnarray*}
Note that 
\begin{eqnarray*}
\supnorm{A_T} \leq \varphi(\lambda) \sum_{t=0}^{T} (1-\varphi(\lambda))^{t}  \supnorm{P^t f-\Pi f}.
\end{eqnarray*}
If we take $\lambda\downarrow{0}$, then the r.h.s. converges to zero. Thus,
\begin{equation*}
\|R_{\lambda}f-\Pi{f}\|_{\infty} \leq \supnorm{A-A_T}, \mbox{ for all } T>0,
\end{equation*}
which concludes the proof.

(b) For any $f\in C_{b}(\cZ)$, we have
\begin{eqnarray*}
\lefteqn{ \|P_{\lambda}^{L}f-Q\Pi{f}\|_{\infty} } \cr
& \leq & \|Q_{\lambda}(R_{\lambda}f-\Pi{f})\|_{\infty} + \|Q_{\lambda}\Pi{f} - Q\Pi{f}\|_{\infty} \cr
& \leq & \|R_{\lambda}f-\Pi{f}\|_{\infty} + \|Q_{\lambda}\Pi{f}-Q\Pi{f}\|_{\infty}.
\end{eqnarray*}
The first term of the r.h.s. approaches 0 as $\lambda\downarrow{0}$ according to (a). The second term of the r.h.s. also approaches 0 as $\lambda\downarrow{0}$ since $Q_{\lambda}\rightarrow{Q}$ as $\lambda\downarrow{0}$.

(c) Note that, by definition of the perturbed t.p.f. $P_{\lambda}$, we have
\begin{equation*}
P_{\lambda}R_{\lambda} = (1-\varphi(\lambda))PR_{\lambda} + \varphi(\lambda)Q_{\lambda}R_{\lambda}.
\end{equation*}
Note further that $Q_{\lambda}R_{\lambda}=P_{\lambda}^{L}$ and $$(1-\varphi(\lambda))PR_{\lambda} = R_{\lambda} - \varphi(\lambda)I,$$ where $I$ corresponds to the identity operator. Thus, we have
\begin{equation*}
P_{\lambda}R_{\lambda} = R_{\lambda}-\varphi(\lambda)I+\varphi(\lambda)P_{\lambda}^{L}.
\end{equation*}
For any invariant probability measure of $P_{\lambda}$, $\mu_{\lambda}$, we have
\begin{equation*}
\mu_{\lambda}P_{\lambda}R_{\lambda} = \mu_{\lambda}R_{\lambda}-\varphi(\lambda)\mu_{\lambda}+\varphi(\lambda)\mu_{\lambda}P_{\lambda}^{L},
\end{equation*}
which equivalently implies that 
\begin{equation*}
\mu_{\lambda} = \mu_{\lambda}P_{\lambda}^{L},
\end{equation*}
since $\mu_{\lambda}P_{\lambda} = \mu_{\lambda}$. Thus, we conclude that $\mu_{\lambda}$ is also an invariant p.m. of $P_{\lambda}^{L}$.

(d) Let $\hat{\mu}$ denote a weak limit point of $\mu_{\lambda}$ as $\lambda\downarrow{0}$. To see that such a limit exists, take $\hat{\mu}$ to be an invariant probability measure of $P$. Then, 
\begin{eqnarray*}
\lefteqn{\|P_{\lambda}f-P{f}\|_{\infty}} \cr & \geq & \|\mu_{\lambda}(P_{\lambda}f-P{f})\|_{\infty} \cr & = & \|(\mu_{\lambda}-\hat{\mu})(I-P)[f]\|_{\infty}.
\end{eqnarray*} 
Note that the weak convergence of $P_{\lambda}$ to $P$, it necessarily implies that $\mu_{\lambda}\Rightarrow\hat{\mu}$. Note further that
\begin{eqnarray*}
\lefteqn{\hat{\mu}[f] - \hat{\mu}Q\Pi{f}}\cr & = & (\hat{\mu}[f]-\mu_{\lambda}[f]) + \mu_{\lambda}[P_{\lambda}^{L}f-Q\Pi{f}]+ \cr && (\mu_{\lambda}[Q\Pi{f}]-\hat{\mu}[Q\Pi{f}]).
\end{eqnarray*}
The first and the third term of the r.h.s. approaches 0 as $\lambda\downarrow{0}$ due to the fact that $\mu_{\lambda}\Rightarrow\hat{\mu}$. The same holds for the second term of the r.h.s. due to part (b). Thus, we conclude that any weak limit point of $\mu_{\lambda}$ as $\lambda\downarrow{0}$ is an invariant p.m. of $Q\Pi$. \QED
\end{proof}

\subsection{Invariant p.m. of one-step perturbed process}	\label{sec:InvariantPMOneStepPerturbedTPF}

Define the finite-state Markov process $\hat{P}$ as in (\ref{eq:FiniteStateMarkovChain}). 

\begin{proposition} [Unique invariant p.m. of $Q\Pi$]		\label{Pr:UniqueInvariantPMofQPi}
There exists a unique invariant probability measure $\hat{\mu}$ of $Q\Pi$. It satisfies 
\begin{equation}	\label{eq:InvariantMeasureQP_derivation}
\hat{\mu}(\cdot) = \sum_{s\in\mathcal{S}}\pi_s\delta_s(\cdot)
\end{equation}
for some constants $\pi_s\geq{0}$, $s\in\mathcal{S}$. Moreover, $\pi=(\pi_1,...,\pi_{\magn{\mathcal{S}}})$ is an invariant distribution of $\hat{P}$, i.e., $\pi=\pi\hat{P}$.
\end{proposition}
\begin{proof}
From Proposition~\ref{Pr:LimitingUnperturbedTPF}(d), we know that the support of $\Pi$ is on the set of pure strategy states $\mathcal{S}$. Thus, the support of $Q\Pi$ is also on $\mathcal{S}$. From Proposition~\ref{Pr:WeakLimitPointsOfPerturbedInvariantMeasures}, we know that $Q\Pi$ admits an invariant measure, say $\hat{\mu}$, whose support is also $\mathcal{S}$. Thus, $\hat{\mu}$ admits the form of (\ref{eq:InvariantMeasureQP_derivation}), for some constants $\pi_{s}\geq{0}$, $s\in\mathcal{S}$.

Note also that $\mathcal{N}_{\delta}(s')$ is a continuity set of $Q\Pi(s,\cdot)$, i.e., $Q\Pi(s,\partial\mathcal{N}_{\delta}(s'))=0$. Thus, from Portmanteau theorem, given that $QP^{t}\Rightarrow Q\Pi$, $$Q\Pi(s,\mathcal{N}_{\delta}(s')) = \lim_{t\to\infty}QP^{t}(s,\mathcal{N}_{\delta}(s')) = \hat{P}_{ss'}.$$  If we also define $\pi_s \df \hat{\mu}(\mathcal{N}_{\delta}(s))$, then
$$\pi_{s'} = \hat{\mu}(\mathcal{N}_{\delta}(s')) = \sum_{s\in\mathcal{S}}\pi_s Q\Pi(s,\mathcal{N}_{\delta}(s')) = \sum_{s\in\mathcal{S}}\pi_s\hat{P}_{ss'},$$ which shows that $\pi$ is an invariant distribution of $\hat{P}$, i.e., $\pi=\pi\hat{P}$.

It remains to establish uniqueness of the invariant distribution of $Q\Pi$. Note that the set $\mathcal{S}$ of pure strategy states is isomorphic with the set $\mathcal{A}$ of action profiles. If agent $i$ trembles (as t.p.f. $Q$ dictates), then all actions in $\mathcal{A}_i$ have positive probability of being selected, i.e., $Q(\alpha,(\alpha_i',\alpha_{-i}))>0$ for all $\alpha_i'\in\mathcal{A}_i$ and $i\in\cI$. It follows by Proposition~\ref{Pr:ConvergenceToPSS} that $Q\Pi(\alpha,(\alpha_i',\alpha_{-i}))>0$ for all $\alpha_i'\in\mathcal{A}_i$ and $i\in\cI$. Finite induction then shows that $(Q\Pi)^{n}(\alpha,\alpha')>0$ for all $\alpha,\alpha'\in\cA$. It follows that if we restrict the domain of $Q\Pi$ to $\mathcal{S}$, it defines an irreducible stochastic matrix. Therefore, $Q\Pi$ has a unique invariant distribution. \QED
\end{proof}

\subsection{Proof of Theorem~\ref{Th:StochasticStability}}	\label{sec:Proof:StochasticStability}

Theorem~\ref{Th:StochasticStability}(a)--(b) is a direct implication of Propositions~\ref{Pr:WeakLimitPointsOfPerturbedInvariantMeasures}--\ref{Pr:UniqueInvariantPMofQPi}.

\section{Conclusions \& Future Work}    \label{sec:Conclusions}

In this paper, we considered a class of reinforcement-learning algorithms that belong to the family of learning automata, and we provided an explicit characterization of the invariant probability measure of its induced Markov chain. Through this analysis, we demonstrated convergence (in a weak sense) to the set of pure-strategy states, overcoming prior restrictions necessary under an ODE-approximation analysis, such as the existence of a potential function. Thus, we opened up new possibilities for equilibrium selection through this type of algorithms that goes beyond the fine class of potential games.

Although the set of pure-strategy-states (which are the stochastically-stable states) may contain non-Nash pure strategy profiles, a follow-up analysis that excludes convergence to such pure-strategy-states may be performed (similarly to the analysis presented in \cite{ChasparisShammaRantzer15} for diminishing step size).

\bibliographystyle{IEEEtran}
\bibliography{2016_SSReinforcementLearning_Bibliography.bib}

\begin{thebibliography}{10}
\providecommand{\url}[1]{#1}
\csname url@rmstyle\endcsname
\providecommand{\newblock}{\relax}
\providecommand{\bibinfo}[2]{#2}
\providecommand\BIBentrySTDinterwordspacing{\spaceskip=0pt\relax}
\providecommand\BIBentryALTinterwordstretchfactor{4}
\providecommand\BIBentryALTinterwordspacing{\spaceskip=\fontdimen2\font plus
\BIBentryALTinterwordstretchfactor\fontdimen3\font minus
  \fontdimen4\font\relax}
\providecommand\BIBforeignlanguage[2]{{%
\expandafter\ifx\csname l@#1\endcsname\relax
\typeout{** WARNING: IEEEtran.bst: No hyphenation pattern has been}%
\typeout{** loaded for the language `#1'. Using the pattern for}%
\typeout{** the default language instead.}%
\else
\language=\csname l@#1\endcsname
\fi
#2}}

\bibitem{Tsetlin73}
M.~Tsetlin, \emph{Automaton Theory and Modeling of Biological Systems}.\hskip
  1em plus 0.5em minus 0.4em\relax Academic Press, 1973.

\bibitem{Narendra89}
K.~Narendra and M.~Thathachar, \emph{Learning Automata: An introduction}.\hskip
  1em plus 0.5em minus 0.4em\relax Prentice-Hall, 1989.

\bibitem{Arthur93}
W.~B. Arthur, ``On designing economic agents that behave like human agents,''
  \emph{J. Evolutionary Econ.}, vol.~3, pp. 1--22, 1993.

\bibitem{BorgersSarin97}
T.~B{\"{o}}rgers and R.~Sarin, ``Learning through reinforcement and replicator
  dynamics,'' \emph{J. Econ. Theory}, vol.~77, no.~1, pp. 1--14, 1997.

\bibitem{Erev98}
I.~Erev and A.~Roth, ``Predicting how people play games: reinforcement learning
  in experimental games with unique, mixed strategy equilibria,'' \emph{Amer.
  Econ. Rev.}, vol.~88, pp. 848--881, 1998.

\bibitem{HopkinsPosch05}
E.~Hopkins and M.~Posch, ``Attainability of boundary points under reinforcement
  learning,'' \emph{Games Econ. Behav.}, vol.~53, pp. 110--125, 2005.

\bibitem{Beggs05}
A.~Beggs, ``On the convergence of reinforcement learning,'' \emph{J. Econ.
  Theory}, vol. 122, pp. 1--36, 2005.

\bibitem{ThathacharSastry04}
M.~Thathachar and P.~Sastry, \emph{Networks of Learning Automata: Techniques
  for Online Stochastic Optimization}.\hskip 1em plus 0.5em minus 0.4em\relax
  Kluwer Academic Publishers, 2004.

\bibitem{ChasparisShamma11_DGA}
G.~Chasparis and J.~Shamma, ``Distributed dynamic reinforcement of efficient
  outcomes in multiagent coordination and network formation,'' \emph{Dynamic
  Games and Applications}, vol.~2, no.~1, pp. 18--50, 2012.

\bibitem{Leslie04}
D.~Leslie, ``Reinforcement learning in games,'' Ph.D. dissertation, School of
  Mathematics, University of Bristol, 2004.

\bibitem{ChasparisShammaRantzer15}
G.~C. Chasparis, J.~S. Shamma, and A.~Rantzer, ``Nonconvergence to saddle
  boundary points under perturbed reinforcement learning,'' \emph{Int. J. Game
  Theory}, vol.~44, no.~3, pp. 667--699, 2015.

\bibitem{Posch97}
M.~Posch, ``Cycling in a stochastic learning algorithm for normal form games,''
  \emph{J. Evolutionary Econ.}, vol.~7, pp. 193--207, 1997.

\bibitem{Sastry94}
P.~Sastry, V.~Phansalkar, and M.~Thathachar, ``Decentralized learning of {N}ash
  equilibria in multi-person stochastic games with incomplete information,''
  \emph{IEEE Trans. Syst. Man Cybern.}, vol.~24, no.~5, pp. 769--777, 1994.

\bibitem{verbeeck_exploring_2007}
K.~Verbeeck, A.~Nowé, J.~Parent, and K.~Tuyls,
  ``\BIBforeignlanguage{en}{Exploring selfish reinforcement learning in
  repeated games with stochastic rewards},''
  \emph{\BIBforeignlanguage{en}{Autonomous Agents and Multi-Agent Systems}},
  vol.~14, no.~3, pp. 239--269, Apr. 2007.

\bibitem{hu_nash_2003}
J.~Hu and M.~P. Wellman, ``Nash {Q}-learning for general-sum stochastic
  games,'' \emph{J. Machine Learning Research}, vol.~4, no. Nov, pp.
  1039--1069, 2003.

\bibitem{littman_markov_1994}
M.~L. Littman, ``Markov games as a framework for multi-agent reinforcement
  learning,'' in \emph{Proc. Int. Conf. Machine Learning}.\hskip 1em plus 0.5em
  minus 0.4em\relax Morgan Kaufmann, 1994, pp. 157--163.

\bibitem{chapman_convergent_2013}
A.~C. Chapman, D.~S. Leslie, A.~Rogers, and N.~R. Jennings,
  ``\BIBforeignlanguage{en}{Convergent {Learning} {Algorithms} for {Unknown}
  {Reward} {Games}},'' \emph{\BIBforeignlanguage{en}{SIAM J. Control Optim.}},
  vol.~51, no.~4, pp. 3154--3180, Jan. 2013.

\bibitem{MondererShapley96}
D.~Monderer and L.~Shapley, ``Potential games,'' \emph{Games Econ. Behav.},
  vol.~14, pp. 124--143, 1996.

\bibitem{leslie_individual_2005}
D.~Leslie and E.~Collins, ``Individual {Q}-{Learning} in {Normal} {Form}
  {Games},'' \emph{SIAM J. Control Optim.}, vol.~44, no.~2, pp. 495--514, Jan.
  2005.

\bibitem{marden_payoff_2009}
J.~R. Marden, H.~P. Young, G.~Arslan, and J.~S. Shamma, ``Payoff based dynamics
  for multi-player weakly acyclic games,'' \emph{SIAM J. Control Optim.},
  vol.~48, no.~1, pp. 373--396, 2009.

\bibitem{ChasparisAriShamma13_SIAM}
G.~Chasparis, A.~Arapostathis, and J.~Shamma, ``Aspiration learning in
  coordination games,'' \emph{SIAM J. Control and Optim.}, vol.~51, no.~1,
  2013.

\bibitem{Lerma03}
O.~Hernandez-Lerma and J.~B. Lasserre, \emph{Markov Chains and Invariant
  Probabilities}.\hskip 1em plus 0.5em minus 0.4em\relax Birkhauser Verlag,
  2003.

\end{thebibliography}

\end{document}